%% file: main.tex
\documentclass[twoside,leqno,twocolumn]{article}

\usepackage[letterpaper]{geometry}

\usepackage{ltexpprt}

\usepackage{amsmath,amssymb,amsfonts}
\usepackage{graphicx}
\usepackage{textcomp}
\usepackage{xcolor}
\usepackage{mathtools}

\usepackage{todonotes}

\usepackage{algorithm}
\usepackage{algpseudocode}

\algtext*{EndWhile}%
\algtext*{EndIf}%
\algtext*{EndFor}%
\algtext*{EndProcedure}%

\usepackage{balance}

\usepackage{fontawesome5}

\usepackage{tikz}
\usepackage{csquotes}
\usetikzlibrary{shapes}

\usepackage{hyperref}

\begin{document}

    \title{\Large Shared-Memory Hierarchical Process Mapping}
    \author{Christian Schulz\thanks{Heidelberg University, Germany, christian.schulz@informatik.uni-heidelberg.de}
    \and Henning Woydt\thanks{Heidelberg University, Germany, henning.woydt@informatik.uni-heidelberg.de, Funded by the Deutsche Forschungsgemeinschaft (DFG, German Research Foundation) -- DFG SCHU 2567/6-1}}

    \date{}

    \maketitle

    \fancyfoot[R]{\scriptsize{Copyright \textcopyright\ 2025 by SIAM\\
    Unauthorized reproduction of this article is prohibited}}

\fancyfoot[R]{\scriptsize{Copyright \textcopyright\ 2025\\
Copyright for this paper is retained by authors}}

    \begin{abstract}
        \small\baselineskip=9pt
        Modern large-scale scientific applications consist of thousands to millions of individual tasks.
        These tasks involve not only computation but also communication with one another.
        Typically, the communication pattern between tasks is sparse and can be determined in advance.
        Such applications are executed on supercomputers, which are often organized in a hierarchical hardware topology, consisting of islands, racks, nodes, and processors, where processing elements reside.
        To ensure efficient workload distribution, tasks must be allocated to processing elements in a way that ensures balanced utilization.
        However, this approach optimizes only the workload, not the communication cost of the application.
        It is straightforward to see that placing groups of tasks that frequently exchange large amounts of data on processing elements located near each other is beneficial.
        The problem of mapping tasks to processing elements considering optimization goals is called process mapping.
        In this work, we focus on minimizing communication cost while evenly distributing work.
        We present the first shared-memory algorithm that utilizes hierarchical multisection to partition the communication model across processing elements.
        Our parallel approach achieves the best solution on 95 percent of instances while also being marginally faster than the next best algorithm.
        Even in a serial setting, it delivers the best solution quality while also outperforming previous serial algorithms in speed.
    \end{abstract}

    \section{Introduction}\label{sec:introduction}
    \input{introduction}

    \section{Preliminaries}\label{sec:preliminaries}
    \input{preliminaries}

    \section{Related Work}\label{sec:related-work}
    \input{related_work}

    \section{Parallel Hierarchical Multisection}\label{sec:parallel-hierarchical-multisection}
    \input{parallel_hierarichal_multisection}

    \section{Achieving a Balanced Partition}\label{sec:balanced-partition}
    \input{balanced_partition}

    \section{Experimental Evaluation}\label{sec:experiments}
    \input{experiments}

    \section{Conclusion}\label{sec:conclusion}
    \input{conclusion}

    \bibliographystyle{plainurl}
    \balance
    \bibliography{bibliography}
\end{document}

%% file: introduction.tex
The performance of applications on modern distributed parallel systems depends on several factors, one of the most critical being the communication between processing elements (PEs).
It is straightforward to see that placing tasks that frequently exchange large amounts of data on PEs that are physically close to one another is far more efficient than placing them on distant PEs\@.
On modern supercomputers, the distance between PEs is captured by the hardware topology and the corresponding communication links.
The hardware topology is typically organized in a hierarchy that includes islands, racks, nodes, and processors.
Consequently, the speed of communication tends to degrade as the distance between PEs increases within the hierarchy.
Often, both the communication pattern between application tasks and the underlying hardware topology are known in advance.
By optimizing the placement of tasks onto the PEs, significant performance improvements can be achieved.
A mapping from $n$ tasks to $k$ PEs is desired that minimizes the cost of communication while also balancing the workload of the tasks onto the PEs.
This optimization problem is known as the \hbox{\emph{general process mapping problem} (GPMP)}.

Sahni and Gonzalez~\cite{Sahni76} showed that solving the underlying \hbox{\emph{quadratic assignment problem} (QAP)} is \textsc{NP}-hard and that no constant factor approximation exists unless $\textsc{P} = \textsc{NP}$.
This inherent difficulty is also evident in practice, as no meaningful instances with $n > 20$ can be solved optimally~\cite{Burkard98}.
Consequently, mapping instances with thousands or millions of tasks can only be tackled through heuristic algorithms.

In this work, we make two assumptions that are generally valid for modern supercomputers and the applications that run on them: (A) the application’s communication pattern is sparse, and (B) the hardware communication topology is hierarchical, with uniform communication speed at the same level of the hierarchy.
Assumption (A) stems from the communication pattern of large-scale scientific applications, see e.g.,~\cite{Catalyurek96, Fietz12, Schloegel00}.
To efficiently distribute the workload and minimize communication cost, \emph{graph partitioning} (GP) is typically deployed, resulting in a sparse communication pattern.
Assumption (B) is typically satisfied by modern supercomputers, which are often built with homogeneous architectures.
Such architectures offer several advantages, e.g.,\ simplicity, scalability and performance consistency.

The two main approaches to solving GPMP are \hbox{(i) integrated-mapping} and \hbox{(ii) the two-phase} approach.
The first approach integrates the mapping directly into the multilevel graph partitioning process, see~\cite{Faraj20, Pellegrini96, Walshaw01}.
Here, the objective function, typically the number of cut edges, is replaced by a function that accounts for processor distances.

The second approach decouples partitioning and mapping into two distinct phases.
In the first phase, a balanced \hbox{$k$-way partition} that minimizes the communication via edge-cut is computed.
In the second phase, each block is mapped to one~PE of the processor network, such that the total communication cost is minimized.
A special variant of this approach is called \emph{hierarchical multisection}~\cite{Schulz19}, which partitions the graph along the systems' hierarchy.
First, the communication graph is partitioned onto the islands, yielding a balanced partition with small edge-cut, i.e., few communications across the islands.
Next, the subgraph of each island is partitioned across the racks within the island, again yielding a small edge-cut partition and therefore few communications across the racks.
This procedure is repeated for each layer in the topology.
In the end, a balanced $k$-way partition is obtained, allowing the trivial identity mapping to solve the mapping phase.
The approach was first presented in~\cite{Schulz19} and experimental evaluations reported in~\cite{Faraj20, Heuer23, Schulz19} show that it produces high-quality mappings.

\textbf{Our Contribution.}
We present \textsc{SharedMap}\footnote{\url{https://github.com/HenningWoydt/SharedMap}}, the first parallel shared-memory hierarchical multisection algorithm to tackle GPMP using the two-phase approach.
Partitioning along the system's hierarchy naturally divides the problem into independent subproblems, providing an intuitive opportunity for parallelization.
Our parallel approach achieves the best solution quality on $95\%$ of instances while being marginally faster than the next best parallel solver.
Furthermore, our serial version achieves better solution quality while also outperforming the previous best serial algorithms in terms of speed.

%% file: preliminaries.tex
\subsection{Concepts and Notation.}\label{subsec:concepts}
Let $G=(V, E)$ be an undirected graph, $c: V \to \mathbb{R}_{\geq 0}$ the vertex weights, and $\omega: E \to \mathbb{R}_{\geq 0}$ the edge weights.
The natural extension of both functions to sets is $c(V') = \sum_{v \in V'} c(v)$ for any $V' \subseteq V$ and $\omega(E') = \sum_{e \in E'} \omega(e)$ for any $E' \subseteq E$.

The \emph{graph partitioning problem} (GPP) asks to partition the vertex set $V$ of a graph into $k$ distinct blocks, usually under a \emph{balancing} constraint.
Formally, GPP partitions $V = V_1 \cup V_2 \cup \ldots \cup V_k$, such that $\bigcup_{i}V_i = V$ and $\forall i\neq j: V_i \cap V_j = \emptyset$, which is called a \hbox{\emph{$k$-way partition}} of $G$.
The balancing constraint is controlled by a hyperparameter $\epsilon \in \mathbb{R}$, called the \emph{imbalance}.
The sum of vertex weights in each partition may not exceed $L_{\max} \coloneqq (1 + \epsilon)\frac{c(V)}{k}$, that is $c(V_i) \leq L_{\max}$ for all~$i$.
The \emph{edge-cut} of a $k$-partition is defined as the total weight of edges that cross between different blocks, i.e., $\sum_{i < j}\omega(E_{ij})$ with $E_{ij} = \{\{u, v\} \mid u \in V_i, v \in V_j\}$.

Let $n$ denote the number of tasks and $k$ the number of PEs.
The communication matrix of the tasks is denoted by $\mathcal{C} \in \mathbb{R}^{n \times n}$ and the hardware topology matrix is denoted by $\mathcal{D} \in \mathbb{R}^{k \times k}$.
An entry $\mathcal{C}_{ij}$ represents the amount of communication between tasks $i$ and $j$ and an entry $\mathcal{D}_{xy}$ denotes the communication factor between PEs $x$ and $y$.
The communication cost is therefore $\mathcal{C}_{ij}\mathcal{D}_{xy}$ if tasks $i$ and $j$ are assigned to PE's~$x$ and $y$ respectively.
Both $\mathcal{C}$ and $\mathcal{D}$ are assumed to be symmetric, otherwise they can be modeled symmetrically~\cite{Brandfass12}.
Throughout this work, we will consider the communication graph $G_{\mathcal{C}}$ instead of $\mathcal{C}$.
The communication graph contains a forward and backward edge with weight $C_{ij}$ between vertices~\hbox{$i$ and $j$} for each non-zero entry $C_{ij}$.
This representation is more beneficial as communication matrices are typically sparse.

In \emph{hierarchical process mapping} the topology of the supercomputer is described by a homogeneous~\hbox{hierarchy $H = a_1 : a_2 : \ldots : a_\ell$}.
This hierarchy specifies that each processor contains $a_1$ PEs, each node contains $a_2$ processors, each rack contains $a_3$ nodes, and so on.
The total number of PEs is $k = \prod_{i=1}^{\ell}a_i$.
Additionally, the sequence $D = d_1 : d_2 : \ldots : d_\ell$ describes the communication cost between the different PEs.
Two PEs on the same processor have distance~$d_1$, two PEs on the same node but on different processor have distance~$d_2$, two PEs in the same rack but on different nodes have distance~$d_3$, and so forth.

The focus of this work is to solve the \emph{general process mapping problem}.
It asks to assign each vertex of the communication graph $G_\mathcal{C}$ to exactly one PE of the communication topology, such that the total communication cost is minimized while satisfying the balancing constraint.
Formally, the objective is to find a mapping $\Pi : [n] \to [k]$ that minimizes \hbox{$J(C, D, \Pi) \coloneqq \sum_{i,j}\mathcal{C}_{ij}\mathcal{D}_{\Pi(i)\Pi(j)}$}.
In GPMP, it is typically assumed that $n > k$.
If $n = k$, the problem is known as the \emph{one-to-one process mapping problem} (OPMP), which is equivalent to QAP\@.
The problems GPP, GPMP and QAP are all \textsc{NP}-hard problems~\cite{Faraj20, Garey76, Sahni76} and no constant factor approximation guarantee exists unless \textsc{P = NP}.

\subsection{Graph Partitioning.}\label{subsec:graph-partitioning}
Since graph partitioning is a crucial component of our approach, we provide a brief overview of the approach and the libraries we employ.
Modern graph partitioning typically relies on the \emph{multilevel} approach.
This approach iteratively coarsens the graph into smaller and smaller graphs while preserving its overall structure.
Once the graph is sufficiently small, a high-quality, albeit potentially expensive, partitioning algorithm is used.
The partition is projected back onto the larger graphs in reverse order, and local search algorithms are employed to further improve the objective function.

The two main methods to coarsen a graph are edge contraction~\cite{Karypis95} and clustering~\cite{Meyerhenke16}.
Contracting an edge~$\{u, v\}$ is achieved by replacing $u$ and $v$ with a new vertex $w$ with weight $c(w) = c(u) + c(v)$.
The former neighbors of $u$ and $v$ are connected to $w$ and parallel edges $\{u, x\}$ and $\{v, x\}$ are replaced by the edge~$\{w, x\}$ with weight \hbox{$\omega(\{w, x\}) = \omega(\{u, x\}) + \omega(\{v, x\})$}.
A common strategy involves determining a matching on the graph and applying edge contraction on all edges in the matching.
Coarsening through clustering works by first determining a set of vertices as clusters and then contracting the neighborhood of each cluster.
The weight of the cluster vertex is the summed weights of the vertices contained in the cluster, and parallel edges are handled in the same way as in edge contraction.

Coarsening is applied iteratively until a graph has fewer vertices than a pre-determined threshold.
For this smallest graph, a balanced $k$-way partition optimizing edge-cut is computed using either recursive bisection or a direct $k$-way partitioning algorithm.
Since the graph is small, more computationally expensive algorithms can be employed to achieve high-quality partitions, e.g.,\ small edge-cut partitions.

The graph is then uncontracted, and local search is applied to further improve the edge-cut.
Local search methods consist of finding vertex moves or series of moves to other blocks that improve the edge-cut without violating the balance constraint.
The most widely used local search methods include the \hbox{Fiduccia-Mattheyses (FM)} algorithm~\cite{Fiduccia82}, \hbox{$k$-way FM~\cite{Sanders11}} and \hbox{Flow-Based Refinement~\cite{Sanders11}}.
Uncontraction and refinement are repeated until the original graph is fully restored.

In this work, we use two libraries for graph partitioning.
For the serial case we use \textsc{KaFFPa}~\cite{Sanders11} from the \textsc{KaHIP} library, and for parallel shared-memory partitioning we use \textsc{Mt-KaHyPar}~\cite{Gottesbueren20}.
\textsc{KaFFPa} is, in general, considered to be one of the best partitioners due to its high quality partitions~\cite{Schlag21}.
\textsc{Mt-KaHyPar} is specialized for hypergraph-partitioning instead of graphs, however, it can also handle those.
Experimental evaluation in~\cite{Gottesbueren20} shows that its quality for graph partitioning is comparable to \textsc{Mt-KaHIP}~\cite{Akhremtsev18}.
Both \textsc{KaFFPa} and \hbox{\textsc{Mt-KaHyPar}} can solve GPMP, which makes them interesting algorithms to compare to in the experiment section~\ref{sec:experiments}.
Consequently, we choose these as serial and parallel graph partitioners.
Section~\ref{sec:related-work} describes both approaches in more detail.

Swapping the libraries with potentially faster and/or stronger alternatives will directly translate to a faster and stronger algorithm for our implementation.
However, both libraries already provide high-quality solutions, meaning any additional gains would likely be marginal.
Other notable serial graph partitioners include \textsc{Jostle}~\cite{Walshaw00}, \textsc{Metis}~\cite{Karypis98}, \textsc{Scotch}~\cite{Pellegrini96} and \textsc{KaHyPar}~\cite{Schlag21}.
For shared-memory graph partitioning, prominent options include \textsc{Mt-KaHIP}~\cite{Akhremtsev18}, \textsc{Mt-Metis}~\cite{Lasalle13}, \textsc{KaMinPar}~\cite{Gottesburen21} and \textsc{PuLP}~\cite{Slota14}.

%% file: related_work.tex
Process mapping is closely related to graph partitioning, a field that has seen a tremendous amount of research.
Section~\ref{subsec:graph-partitioning} already gave a brief introduction, and we refer the reader to~\cite{bichot13} and~\cite{Buluc16} for more information.

GPMP has likewise seen large amounts of research, and we refer the reader to~\cite{Hoefler13} and~\cite{Pellegrini13} for a more detailed overview.
Hatazaki~\cite{Hatazaki98} was among the first to use graph partitioning to map MPI processes onto hardware topologies.
Träff~\cite{Traff02} implemented a non-trivial MPI mapping for the NEC SX-series of parallel vector computers and Yu et al.~\cite{Yu06} implemented graph embedding for the Blue Gene/L Supercomputer.

OPMP, the second phase of the two-phase approach, likewise has seen a lot of research.
Müller-Merbach~\cite{MullerMerbach70} presented a greedy algorithm to obtain an initial mapping for OPMP\@, that serves as a basis for many subsequently developed heuristics.
An improvement that maintains the same asymptotic complexities is offered by Glantz et al.~\cite{Glantz15}.
Heider~\cite{Heider72} introduced a refinement step that iteratively considers all $\mathcal{O}(n^2)$ possible swaps in the mapping.
Brandfass et al.~\cite{Brandfass12} improve the method by avoiding redundant swaps and also propose a method to split the mapping into multiple blocks and only perform swaps inside each block.
The method is further refined by Schulz and Träff~\cite{Schulz17} by using more efficient data structures to compute the gains of swaps and further restricting the search space.

While optimizing $J(C, D, \Pi)$ seems straightforward from a theoretical perspective, it raises the question of whether optimizing this metric leads to improvement in practice.
Brandfass et al.~\cite{Brandfass12} have shown that optimizing $J(C, D, \Pi)$ in the context of Computational Fluid Dynamics (CFD) leads to a practical improvement.
Alternative objective functions also have been considered.
For example, Hoefler and Snir~\cite{Hoefler11} minimize the maximum network congestion.

Next, we present five state-of-the-art algorithms in more detail, as we use them as benchmarking algorithms in Section~\ref{sec:experiments}.
The first three are serial algorithms, the fourth is a shared-memory algorithm, and the last algorithm is a distributed-memory algorithm.

\textsc{KaFFPa-Map}~\cite{Schulz17} solves GPMP using the two-phase approach.
In the first phase $G_\mathcal{C}$ is partitioned into $k$ blocks using recursive bisection, resulting in a communication model graph $G_{M}$.
This new graph contains exactly $k$ vertices and an edge between vertices exist if the $k$-way partition has edges between the corresponding blocks, i.e., $G_{M}$ is the quotient graph.
Edge weights of $G_{M}$ correspond to the summed edge weights between the blocks.
In the second phase, the systems hierarchy~$H$ is used to first create a perfectly balanced \hbox{$a_{\ell}$-way} partition of $G_{M}$, then each of the subgraphs is perfectly partitioned into $a_{\ell - 1}$ blocks and so forth.
The last step will end in $a_1$ partitions, and the mapping along the partitioning is used to map the $k$ blocks onto the $k$ PEs.
To refine the mapping, the local search of~\cite{Brandfass12} is further improved and optimized.
The search consists of swapping the assignment of two processes and evaluating if the objective function is improved.
While~\cite{Brandfass12} allowed swapping of all processes,~\cite{Schulz17} only allows swapping of processes if their distance in $G_{\mathcal{C}}$ does not exceed a threshold $d$.
They showed that $d=10$ achieves a good trade-off between solution quality and runtime.

\textsc{Global Multisection}~\cite{Schulz19} works very similarly to \textsc{KaFFPa-Map}.
Instead of creating the communication model graph $G_{M}$ and mapping it, the communication graph $G_{\mathcal{C}}$ is directly partitioned along the \hbox{hierarchy~$H$}.
The resulting blocks are mapped to the PEs using the identity mapping.
The local search of \textsc{KaFFPa-Map} is applied to refine the mapping.

\textsc{Integrated Mapping}~\cite{Faraj20} abandons the two-phase approach and integrates the mapping into the partitioning of $G_{\mathcal{C}}$.
Rather than optimizing the edge-cut, the communication cost~$J(C, D, \Pi)$ is used as the objective to minimize.
The partitioning follows the multilevel scheme, coarsening using edge contraction, computing an initial solution and uncontracting using local refinement.
To compute the initial solution, the hierarchical multisection approach is applied on the coarsest graph.
Local refinement is performed using a combination of techniques, including quotient-graph refinement, \hbox{$k$-way FM}, label propagation, and multi-try FM\@.
Additionally, four methods were developed to query the distance between two PEs (a frequently needed operation), offering trade-offs between runtime and memory consumption.
Delta-Gain updates are introduced that can, in some cases, reduce the cost of recomputing gains of vertex moves.

\textsc{Mt-KaHyPar-Steiner}~\cite{Heuer23} provides a parallel shared-memory algorithm to solve GPMP by optimizing the Steiner Tree Metric.
Originally,~\cite{Heuer23} deals with minimizing the wire length in VLSI designs.
The goal is to map the logical units of a circuit onto a physical layout (the chip), such that the length of necessary wire is minimized.
The logical units and their connections can be represented by a hypergraph $\mathcal{H}$ while the physical layout can be represented by an edge-weighted graph $\mathcal{C}$.
The edge weights represent the distance of two places on the chip.
The goal is to place the logical units (vertices of $\mathcal{H}$) onto the places of the chip (vertices of $\mathcal{C}$) such that the total wire length spanned by the hyperedges of $\mathcal{H}$ on $\mathcal{C}$ is minimized.
By substituting $\mathcal{H}$ with the communication (hyper-)graph $G_{\mathcal{C}}$ and $\mathcal{C}$ by a complete graph that models the topology matrix~$\mathcal{D}$, the same method can be used to solve GPMP\@.

The method follows the multilevel partitioning scheme for hypergraphs, i.e., the hypergraph is coarsened to a smaller size, an initial mapping is determined and uncontractions with local refinements are performed.
Coarsening is achieved via clustering, and initial solutions are computed by obtaining a balanced $k$-way partition and constructing a mapping onto $C$ in a greedy manner.
Refinement techniques include label propagation, FM-local search and flow-based search.

\textsc{ParHIPMap}~\cite{Predari21} introduced a distributed-memory algorithm to solve GPMP via MPI\@.
It leverages the distributed-memory partitioner \textsc{ParHIP}~\cite{Meyerhenke17} by integrating the mapping into the partitioning, i.e., instead of minimizing edge-cut the communication cost is minimized.
\textsc{ParHIP} also follows the multilevel scheme for partitioning, but each phase is modified due to the distributed approach.
Initially, the graph is distributed across the PEs using block partitioning.
Each PE coarsens its subgraph, the resulting coarsened subgraphs are collected by all PEs, and each PE calculates a mapping on the coarsest graph.
The best mapping is broadcast to all PEs, after which uncoarsening and refinement is performed.
We refer the reader to~\cite{Meyerhenke17} and~\cite{Sanders12} for more details on distributed graph partitioning.
To query the distance of PEs, a bit-label technique is employed, which encodes the ancestors of a PE in one machine word.
This enables them to query the distance of PEs in $\mathcal{O}(1)$ (if hardware instructions are available), but of course also restricts the bit label size and therefore the size of $k$.
\textsc{ParHIPMap} utilizes parallel label propagation as a refinement process.
They additionally avoid high memory usage during block partitioning by removing ``Halo Hubs'', vertices with high degree and edges that span across multiple PEs.

%% file: parallel_hierarichal_multisection.tex
Hierarchical multisection, first introduced in~\cite{Schulz19}, exploits the given hierarchy $H=a_1:a_2:\ldots:a_{\ell}$ of the supercomputer.
Instead of arbitrarily partitioning the communication graph into a balanced $k$-way partition, the graph is partitioned according to the hierarchy.
The communication graph $G_{\mathcal{C}}$ is first partitioned into $a_{\ell}$ blocks.
Each of these blocks is then partitioned further into $a_{\ell - 1}$ blocks and so forth until $k$ blocks are obtained.
The following mapping phase can then be solved trivially by the identity mapping.
Experimental evaluations in ~\cite{Faraj20, Heuer23} and~\cite{Schulz19} show that the approach generates high-quality mappings.
Fig.~\ref{fig:multisection} shows an example.

Besides minimizing communication cost, an \hbox{$\epsilon$-balanced} $k$-way partition is desired.
To achieve such a partition, it is necessary to adaptively rescale $\epsilon$ for each partitioning of a subgraph.
This rescaling is based on the subgraph's weight, the remaining partitions, and its position within the hierarchy.
In the following sections, we denote the adaptive imbalance by $\epsilon'$.
Section~\ref{sec:balanced-partition} provides further detail on the necessity.

\begin{figure*}
    \centering
    \input{multisection}
    \caption{The hierarchical multisection approach with hierarchy $H=4:2:3$ and $D=1:10:100$.
    On the left-hand side, the partitioning of $G_{\mathcal{C}}$ into $k = 4\cdot 2 \cdot 3 = 24$ blocks is shown.
    First $G_{\mathcal{C}}$ is partitioned into three blocks~\hbox{($G_{1}^2$, $G_{2}^2$, $G_{3}^2$)}, each of the blocks is further partitioned into two blocks ($G_{1}^1$ to $G_{6}^1$), and finally, each of these is partitioned into four blocks ($G_1$ to $G_{24}$).
    On the right side, the resulting partitioning and the corresponding communication graph are depicted.
    Solid lines indicate a communication factor of 1 between communicating tasks, dashed lines indicate a factor of 10, and the dotted lines indicate a factor of 100.
    For example, if a task in $G_1$ communicates with a task in $G_4$, the cost is scaled by 1.
    If it communicates with a task in $G_6$ the cost is scaled by 10 and if it communicates with a task in $G_{14}$ the cost is scaled by 100.
    }
    \label{fig:multisection}
\end{figure*}
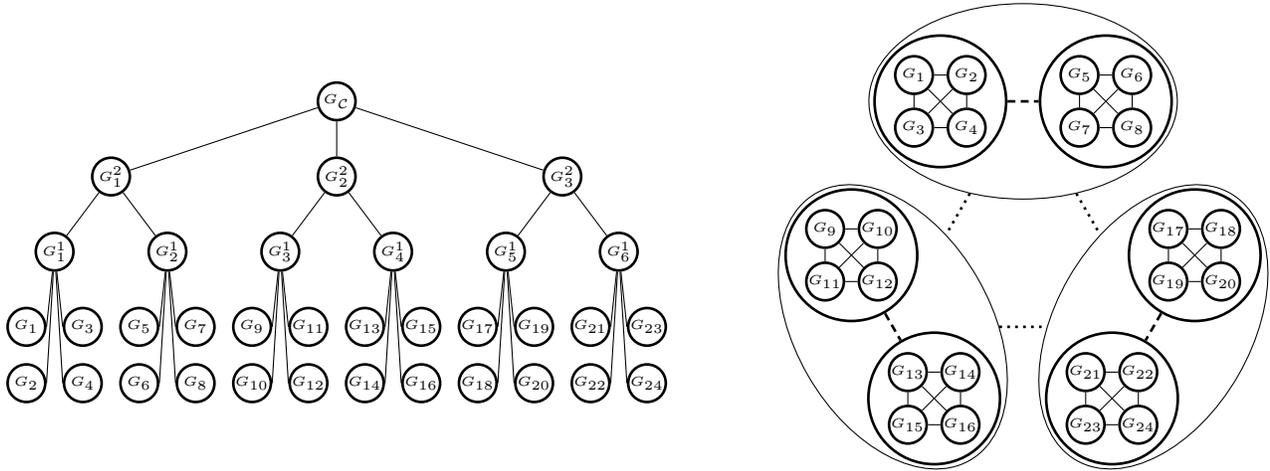

\subsection{Parallel Model.}\label{subsec:parallel-model}
In our parallel model, we assume the availability of $p$ threads that share memory, so no explicit communication is required.
Ideally, a parallelized algorithm would achieve a linear speedup, meaning that using $p$ threads would reduce the runtime to approximately~$1/p$ of the serial runtime.
However, achieving such an ideal speedup is generally not possible due to several factors, one of the most significant being thread idleness.
For an efficient approach, it is (in most cases) crucial that as few threads as possible are idle at all time.

The hierarchical multisection approach naturally offers itself up for parallelization.
Each subgraph is partitioned into multiple blocks, which can then be partitioned independently.
Initially, ~$G_\mathcal{C}$ is partitioned, and all $p$ threads are used.
The partitioning creates~$a_{\ell}$~subgraphs that all have to be partitioned using the $p$ threads.
The question is: How do we distribute the threads among the subgraphs, such that as few threads as possible are idle throughout the whole algorithm?
Note that we will consider threads active if we assign them to partition a graph.
In general, using all threads during graph partitioning is also not trivial, so in practice they may not be continuously active.

Here, we present four strategies to achieve efficient thread distribution.
The first one is straightforward and distributes all threads to one graph.
The second approach iteratively processes each layer of the hierarchy, i.e., first all islands, then all racks, and so forth.
The threads are distributed across all available subgraphs of the layer.
The third one uses a priority queue to collect available graphs and distribute the available threads.
The last approach is similar to the first one, but instead of globally managing graphs and threads, it will locally distribute a fraction of the threads on a fraction of available graphs.

\subsection{\textsc{Naive}.}\label{subsec:naive}
The \textsc{Naive} approach does not distribute threads across all available graphs.
Instead, it simply uses all $p$ threads to partition one graph at a time.
While this approach achieves an optimal distribution scheme, where no threads are idle, it suffers from the fact that graph partitioning does not scale optimally~\cite{Akhremtsev18, Gottesbueren20, Lasalle13}, especially if the subgraphs are small.

\subsection{\textsc{Layer}.}\label{subsec:layer}
The \textsc{Layer} approach processes each layer of the hierarchy in parallel, waiting for all subgraphs to be partitioned before advancing.
For example, the first layer only contains $G_\mathcal{C}$, so all available threads are used to partition the graph into $a_{\ell}$ blocks.
The second layer, corresponding to the islands, contains $a_{\ell}$ graphs that all need to be partitioned into $a_{\ell -1}$ blocks, and the third layer, corresponding to the racks, would then contain $a_{\ell} \cdot a_{\ell - 1}$ graphs that need to be partitioned.
Only after partitioning all graphs of one layer, the next one will be processed.
Since all graphs have roughly the same size, an equal thread distribution is desired.

Let $S = \{G_1, \ldots, G_m\}$ be the $m = a_{\ell} \cdot a_{\ell - 1} \cdot \ldots \cdot a_{\ell - i}$ graphs to be partitioned on layer $i + 1$.
For an equal thread distribution, each graph $G_j$ is assigned
\begin{equation}
    \label{eq:layer-distribution}
    p_j =
    \begin{cases}
        \lfloor \frac{p}{m} \rfloor + (j-1 < (p -\lfloor \frac{p}{m} \rfloor m), & \text{if } p \geq m \\
        1 , & \text{else }
    \end{cases}
\end{equation}
threads (with $<$ evaluating to 0 or 1).
If $p \geq m$ (more threads than graphs), each graph is assigned $\lfloor \frac{p}{m} \rfloor$ threads and the remaining $p - \lfloor \frac{p}{m} \rfloor m$ threads are distributed among the first graphs.
If more graphs need to be partitioned than threads are available~($p < m$) each graph is assigned exactly one thread.

\begin{algorithm}[t]
    \caption{The \textsc{Layer} algorithm}\label{alg:layer}
    \begin{algorithmic}[1]
        \Require Graph $G$, Hierarchy $H$=$a_1$:$\ldots$:$a_{\ell}$, Threads $p$
        \State $S \gets \{G\}$ and $N \gets \emptyset$
        \For{$i = 1$ \textbf{to} $\ell$}
            \For{$j = 1$ \textbf{to} $|S|$}
                \textbf{in parallel}
                \State $p_j \gets \text{Use Equation \ref{eq:layer-distribution}}$
                \State $\epsilon' \gets \text{Use Equation \ref{eq:balanced_epsilon}}$
                \State $T \gets \ $partition$(G_j, a_i, p_j, \epsilon')$
                \State $N \gets N \cup T$
            \EndFor
            \State $S \gets N$ and $N \gets \emptyset$
        \EndFor
        \State \textbf{Return} $S$
    \end{algorithmic}
\end{algorithm}

However, only $p$ threads can be active at any given time.
If~$p < m$, which is common at deeper levels, it is not feasible to start all~$m$ partitionings at once.
Doing so would oversubscribe the available hardware, leading to degraded performance.
To circumvent this problem, only the first $p$ partitionings are started, and additionally, an atomic index $i$ is used that points to the next graph in the set that needs to be partitioned.
When a thread finishes partitioning, it will read and increment the index $i$ in one atomic operation.
If the read index points to a valid graph ($i \leq m$), the thread will partition that graph.
Otherwise, no more graphs remain to be partitioned, and the thread is terminated.
It is crucial that read-and-increment operation is performed atomically, otherwise multiple threads could process the same graph.

Algorithm~\ref{alg:layer} shows pseudocode for this approach.
For readability, the atomic index $i$ and the corresponding operations are omitted (another approach will shortly show a similar technique).
The set $S$ contains all graphs in the current layer, while $N$ holds the resulting graphs of the next layer.
The outer loop (line 2) iterates across all hierarchical layers.
The inner loop (line 3) is executed in parallel, meaning that each loop iteration is handled by a different thread.
While the pseudocode shows $|S|$ iterations, in practice, only $\min(p, |S|)$ iterations are launched to prevent oversubscription.
The number of assigned threads $p_j$ is determined with \hbox{index~$j$} and $m = |S|$ in line 4.
In line 5 the adaptive imbalance $\epsilon'$ is computed, such that an overall \hbox{$\epsilon$-balanced} partition remains possible (see Section~\ref{sec:balanced-partition} for details).
The partitioning takes the subgraph~$G_j$, the number of \hbox{blocks $a_i$}, the number of \hbox{threads $p_j$} and the adaptive \hbox{imbalance $\epsilon'$} and returns a set~$T$ of $a_i$~subgraphs.
Since all graphs in $S$ are partitioned into $a_i$~blocks, it follows that $|N| = |S|a_i$.
The necessary space for $N$ can be pre-allocated, and all threads can insert their individual $T$ into $N$ without requiring synchronization.
Once all layers have been processed, the algorithm returns the set $S$ containing $k$ blocks.

The advantage of the \textsc{Layer} approach is its minimal synchronization overhead.
The number of threads assigned to each graph depends only on index $j$, and only if $m > p$ do the threads access the same atomic variable $i$.
However, the approach has an obvious drawback: if one partitioning takes much longer than any other, it will result in idle threads and wasted time.
The next two approaches both try to improve on this shortcoming.

\subsection{\textsc{Priority Queue}.}\label{subsec:priority-queue}
The \textsc{Priority Queue} approach is more flexible than the \textsc{Layer} approach, but it introduces more synchronization between the threads.
Instead of collecting the graphs of each layer in an orderly fashion, the graphs are pushed in a priority queue, which will be processed if threads are available.
The priority queue is ordered by the size of the graphs, such that the largest graph is always at the top of the queue.
Initially, $G_\mathcal{C}$ is partitioned into $a_{\ell}$~blocks using $p$~threads.
The resulting blocks are then inserted into the queue.
A~master-thread is used to spawn new partitioning tasks.
If graphs are in the queue and threads are available, the master-thread assigns some of the available threads to partition the top graph.
That graph is removed from the queue, and the master-thread spawns new partitioning tasks as needed.
A partitioning task consists of partitioning the assigned graph with the assigned number of threads, and additionally placing the blocks back in the queue.

\setlength{\textfloatsep}{5pt} %
\begin{algorithm}[t]
    \caption{The \textsc{Priority Queue} algorithm}\label{alg:priority-queue}
    \begin{algorithmic}[1]
        \Require Graph $G$, Hierarchy $H$, Threads $p$
        \State $Q \gets \{(G, a_{\ell})\}$ and $p_A \gets p$ and $S_{sol} \gets \emptyset$
        \While{not $(Q = \emptyset$ and $p_A = p)$}
            \While{$Q = \emptyset$ or $p_A = 0$}
                \textbf{wait}
            \EndWhile
            \State $p_t \gets \lceil \frac{p_A}{|Q|} \rceil$ \hfill \faLock
            \State $(G_t, a_t) \gets Q.\text{popTop()}$ \hfill \faLock
            \State $p_A \gets p_A - p_t$ \hfill \faLockOpen
            \Procedure{Spawn Thread}{}
                \State $\epsilon' \gets \text{use Equation \ref{eq:balanced_epsilon}}$
                \State $T \gets \ $partition$(G_t, a_t, p_t, \epsilon')$
                \If{$t = 1$}
                    \hfill \faLock
                    \State $S_{sol} \gets S_{sol} \cup T$ \hfill \faLock
                \Else
                    \hfill \faLock
                    \State $Q$.push($T$) \hfill \faLock
                \EndIf
                \State $p_A \gets p_A + p_t$ \hfill \faLockOpen
            \EndProcedure
        \EndWhile
        \State \textbf{Return} S
    \end{algorithmic}
\end{algorithm}

Algorithm~\ref{alg:priority-queue} shows the corresponding pseudocode.
The priority queue $Q$ is initialized with $(G, a_{\ell})$ (the graph and the number of partitions), the number of available threads $p_A$ is initialized to $p$ and the final set~$S$ is initially empty (line 1).
The while-loop \hbox{(line 2)} continues until $Q$ is empty and all threads are available, indicating that all $k$ blocks have been obtained.
In line~3, the master-thread waits until graphs are placed in the queue or threads become available.
When at least one graph and one thread are available, a partitioning task is generated.
The task uses $p_t$ threads and the first graph of the queue $G_t$.
The graph is removed from the queue and~$p_A$ is updated by subtracting $p_t$.
The symbols \faLock\ and \faLockOpen\ indicate that these operations are enclosed by a lock.
Before executing line 4, the lock must be obtained, and after line 6 the lock is released.
This ensures that modifications to $Q$ and $p_A$ are thread-safe.
Lines 7--14 describe the partitioning task, which is spawned by the master-thread.
The master-thread will continue at line 2, while the partitioning task is executed by the $p_t$ spawned threads.
The task includes partitioning the graph in line 9, placing the blocks back in the queue (line 13) and adding the $p_t$ threads back to the pool (line 14).
If $G_t$ belongs to the last layer of the hierarchy, the blocks are not placed in $Q$, but in the final set $S$.
The same lock surrounds lines 10--14 as lines 4--6, as both regions modify $Q$ and $p_A$.
Once the master-thread exits the while-loop at line 2, all $k$ blocks have been obtained, and the set $S$ is returned.

The master-thread always spawns only one partitioning task at a time, as other partitioning tasks could finish during scheduling and therefore add new graphs and free used threads.
This is especially important if a partitioning task works on a graph of the last layer.
The obtained blocks will be inserted into~$S$, instead of~$Q$, but the threads will be made available.
The master-thread can schedule the now freed threads on the remaining graphs in the queue.

Removal and insertion (lines 5 and 13) from and into the priority queue $Q$ are hard to realize in a non-blocking way.
Here we used a simple lock-based approach.
Each time a thread wants to modify the queue, it first has to acquire a lock.
If the lock is open, the thread closes it, modifies the queue and then unlocks it.
If the lock is closed, the thread will wait until that lock is opened by another thread.
This guarantees that only one thread is modifying the queue.
We also put all operations containing $p_A$ in the locked region.
In total, the same lock surrounds lines~\hbox{4--6} and lines~\hbox{10--14}.

This approach mitigates thread idleness as graphs are always pushed in the queue to be processed.
One partitioning taking too long will not stall all other threads from continuing working.
However, the approach has a much larger synchronization overhead as the access to the queue has to be managed.
Our approach via a lock can result in idle threads, as they wait for the lock to be released.

\subsection{\textsc{Non-Blocking Layer}.}\label{subsec:non-blocking-layer}
The \textsc{Non-Blocking Layer} approach is a compromise between both previous approaches.
It aims to minimize thread idleness, similar to \textsc{Priority Queue}, while maintaining a smaller synchronization overhead, similar to \textsc{Layer}.

In the \textsc{Layer} approach, we partition the graphs $G_1, \ldots, G_m$ with $p_1, \ldots, p_m$ threads.
The resulting blocks of all partitionings are collected and all spawned threads terminate.
However, if one partitioning takes too long, it prevents us from continuing to the next layer.
This can be avoided if the blocks are not collected globally, but only locally.
Each time the set of threads~$p_j$ finishes partitioning a graph, the resulting blocks are saved in a set only accessible to those~$p_j$ threads.
Once no more graphs need to be partitioned, the $p_j$~threads can process their blocks, while other threads may still be processing graphs in the current layer.
The $p_j$~threads will then be divided on to the blocks like in the \textsc{Layer} approach.

Algorithm~\ref{alg:layer-nb} presents the corresponding pseudocode.
It takes as input a set of graphs $S$ (initially $\{G\}$), an atomic index $i$ and the thread count $p$.
The local set $R$ stores blocks generated by partitioning with the $p$ threads.
The index $j$ points to a graph in $S$ and is set via an atomic fetch-and-add operation.
It is necessary since multiple other threads could be processing the set $S$ currently, which becomes clear shortly.
The while loop (lines \hbox{2--7}) processes graphs of $S$ and adds the resulting blocks to $R$.
If this is the last layer in the hierarchy, we will add all resulting sets to the final solution $S_{sol}$, free the $p$ threads and return.
The threads are added back to a global thread pool $p_A$, which is an atomic integer indicating how many threads are currently idle.
The first thread to access the thread pool in line 3 gets all currently available threads.
The last four lines set up threads to process the new set $R$.
In total $m = \min\{p, |R|\}$ new threads, holding $p_1, \ldots, p_m$ threads each, are spawned.
The distribution is done via equation~\ref{eq:layer-distribution}.
Each thread recursively calls the same function with parameters $(R, j, p_i)$, sharing the atomic index $j$.
Therefore, they will not process the same graphs of $R$ during their calls.

\setlength{\textfloatsep}{5pt} %
\begin{algorithm}[t]
    \caption{The \textsc{Non-Blocking Layer} algorithm}\label{alg:layer-nb}
    \begin{algorithmic}[1]
        \Require Graph-Set $S$, Atomic Index $i$, Threads $p$
        \State $R \gets \emptyset$ and $j \gets $ AtomicFetchAdd$(i, 1)$
        \While{$j < |S|$}
            \State $p \gets p\ + $ AtomicExchange$(p_A, 0)$
            \State $\epsilon' \gets \text{use Equation \ref{eq:balanced_epsilon}}$
            \State $T \gets $ partition$(S_j, a_i, p, \epsilon')$
            \State $R \gets R \cup T$
            \State $j \gets $ AtomicFetchAdd$(i, 1)$
        \EndWhile
        \If{$\text{isOnLastLayer()}$}
            \State $S_{sol} \gets S_{sol} \cup R$, AtomicAdd$(p_A, p)$, \textbf{Return}
        \EndIf
        \State Atomic index $j = 0$
        \For{$k = 0$ \textbf{to} $\min\{p, |R|\}$}
            \textbf{in parallel}
            \State $p_k \gets $ use equation~\ref{eq:layer-distribution}
            \State \textbf{Spawn thread} with recursive call $(R, j, p_k)$
        \EndFor
    \end{algorithmic}
\end{algorithm}

Each recursive call processes only a local part of the hierarchy, so a delay of a partitioning affects only that region.
By limiting the number of spawned threads to $\min\{p, |R|\}$ we prevent hardware oversubscription and releasing unneeded threads minimizes undersubscription.
Synchronization occurs for the atomic indices and also for the global thread pool.
The atomic indices produce only a small synchronization overhead, since few threads access each atomic index.
The synchronization overhead of the thread pool in contrast is higher, as it is globally accessible by all threads.
However, modern hardware optimizes atomic operations, making them generally faster than the locks used in the \hbox{\textsc{Priority Queue}} approach.

\subsection{Drawbacks.}\label{subsec:drawbacks}
For each approach (excluding \textsc{Naive}), it is possible to construct scenarios that lead to a suboptimal distribution scheme.
For example, the \textsc{Layer} approach suffers from the global collecting into the set $N$.
If all but one partitioning task finishes, the remaining threads remain idle.
The \textsc{Priority Queue} approach could schedule a large graph with one thread, right before a lot of other partitioning tasks finish.
This leads to (1) many idle threads and (2) slow partitioning of the large graph, as only one thread is processing it.
The \hbox{\textsc{Non-Blocking Layer}}, while addressing some drawbacks, shares a general downside that limits all presented approaches.

Consider $a_{\ell} = 2$, we partition $G_\mathcal{C}$ into $G_1$ and $G_2$.
Assume that $G_1$ is hard to partition while $G_2$ is easy to partition.
The varying difficulty can occur, since the number of edges can vary significantly between subgraphs.
All presented approaches use $p/2$ threads to partition $G_1$ and $G_2$.
Since $G_2$ is easy to partition, we will quickly obtain the final~$k/2$ partitions from $G_2$, however, $G_1$ might still be processed.
No matter which approach is used, the $p/2$ threads used for $G_2$ will remain idle until $G_1$ is partitioned.

In general, one partitioning that takes too long will stall the complete process, as no new subgraphs are generated and the idle threads cannot be utilized.
This could be improved by having a graph partitioner that could dynamically receive more threads.
Any idle thread could then be added to a currently running partitioning task.
However, no graph partitioning library supports such a feature.

%% file: multisection.tex
\let\dist\undefined
\newlength{\dist}
\setlength{\dist}{0.5cm}

\tikzstyle{box} = [circle, font=\tiny, draw=black, align=center, line width=1, minimum size=\dist, inner sep=0pt]
\tikzstyle{noBorderBox} = [rectangle, minimum height=\dist, minimum width=\dist, font=\tiny, align=center, line width=1]
\tikzstyle{default_edge} = []
\tikzstyle{edgedot} = [dotted, line width=1pt]
\tikzstyle{edgedashed} = [densely dashed, line width=1pt]

\begin{tikzpicture}

\node at (0cm,0cm) [noBorderBox, minimum height=4.5cm, minimum width=9cm] {};
\node at (9cm,0cm) [noBorderBox, minimum height=4.5cm, minimum width=9cm] {};

\node (GC) at (-0.125cm,2cm) [box] {$G_\mathcal{C}$};

\node (GC12) at (-3.125cm,1cm) [box] {$G_{1}^2$};
\node (GC22) at (-0.125cm,1cm) [box] {$G_{2}^2$};
\node (GC32) at (2.875cm,1cm) [box] {$G_{3}^2$};

\draw [default_edge] (GC) -> (GC12);
\draw [default_edge] (GC) -> (GC22);
\draw [default_edge] (GC) -> (GC32);

\node (GC11) at (-3.875cm,0cm) [box] {$G_{1}^1$};
\node (GC21) at (-2.375cm,0cm) [box] {$G_{2}^1$};
\node (GC31) at (-0.875cm,0cm) [box] {$G_{3}^1$};
\node (GC41) at (0.625cm,0cm) [box] {$G_{4}^1$};
\node (GC51) at (2.125cm,0cm) [box] {$G_{5}^1$};
\node (GC61) at (3.625cm,0cm) [box] {$G_{6}^1$};

\draw [default_edge] (GC12) -> (GC11);
\draw [default_edge] (GC12) -> (GC21);
\draw [default_edge] (GC22) -> (GC31);
\draw [default_edge] (GC22) -> (GC41);
\draw [default_edge] (GC32) -> (GC51);
\draw [default_edge] (GC32) -> (GC61);

\node (G1) at (-4.25cm,-1.0cm) [box] {$G_{1}$};
\node (G2) at (-4.25cm,-1.75cm) [box] {$G_{2}$};
\node (G3) at (-3.5cm,-1.0cm) [box] {$G_{3}$};
\node (G4) at (-3.5cm,-1.75cm) [box] {$G_{4}$};

\node (G5) at (-2.75cm,-1.0cm) [box] {$G_{5}$};
\node (G6) at (-2.75cm,-1.75cm) [box] {$G_{6}$};
\node (G7) at (-2.00cm,-1.0cm) [box] {$G_{7}$};
\node (G8) at (-2.00cm,-1.75cm) [box] {$G_{8}$};

\node (G9) at (-1.25cm,-1.0cm) [box] {$G_{9}$};
\node (G10) at (-1.25cm,-1.75cm) [box] {$G_{10}$};
\node (G11) at (-0.5cm,-1.0cm) [box] {$G_{11}$};
\node (G12) at (-0.5cm,-1.75cm) [box] {$G_{12}$};

\node (G13) at (0.25cm,-1.0cm) [box] {$G_{13}$};
\node (G14) at (0.25cm,-1.75cm) [box] {$G_{14}$};
\node (G15) at (1.00cm,-1.0cm) [box] {$G_{15}$};
\node (G16) at (1.00cm,-1.75cm) [box] {$G_{16}$};

\node (G17) at (1.75cm,-1.0cm) [box] {$G_{17}$};
\node (G18) at (1.75cm,-1.75cm) [box] {$G_{18}$};
\node (G19) at (2.5cm,-1.0cm) [box] {$G_{19}$};
\node (G20) at (2.5cm,-1.75cm) [box] {$G_{20}$};

\node (G21) at (3.25cm,-1.0cm) [box] {$G_{21}$};
\node (G22) at (3.25cm,-1.75cm) [box] {$G_{22}$};
\node (G23) at (4.00cm,-1.0cm) [box] {$G_{23}$};
\node (G24) at (4.00cm,-1.75cm) [box] {$G_{24}$};

\draw [default_edge] (GC11) -> (G1.east);
\draw [default_edge] (GC11) -> (G2.east);
\draw [default_edge] (GC11) -> (G3.west);
\draw [default_edge] (GC11) -> (G4.west);

\draw [default_edge] (GC21) -> (G5.east);
\draw [default_edge] (GC21) -> (G6.east);
\draw [default_edge] (GC21) -> (G7.west);
\draw [default_edge] (GC21) -> (G8.west);

\draw [default_edge] (GC31) -> (G9.east);
\draw [default_edge] (GC31) -> (G10.east);
\draw [default_edge] (GC31) -> (G11.west);
\draw [default_edge] (GC31) -> (G12.west);

\draw [default_edge] (GC41) -> (G13.east);
\draw [default_edge] (GC41) -> (G14.east);
\draw [default_edge] (GC41) -> (G15.west);
\draw [default_edge] (GC41) -> (G16.west);

\draw [default_edge] (GC51) -> (G17.east);
\draw [default_edge] (GC51) -> (G18.east);
\draw [default_edge] (GC51) -> (G19.west);
\draw [default_edge] (GC51) -> (G20.west);

\draw [default_edge] (GC61) -> (G21.east);
\draw [default_edge] (GC61) -> (G22.east);
\draw [default_edge] (GC61) -> (G23.west);
\draw [default_edge] (GC61) -> (G24.west);

\let\radius\undefined
\newlength{\radius}
\setlength{\radius}{2.00cm}

\let\radiussmall\undefined
\newlength{\radiussmall}
\setlength{\radiussmall}{1.10cm}

\let\r\undefined
\newlength{\r}
\setlength{\r}{1.75cm}
\node[draw, ellipse, minimum width=4.1cm, minimum height=2.6cm, rotate=000] (ETop)   at ({9cm + \radius*cos(90)}, {\radius*sin(90)}) {};
\node[draw, ellipse, minimum width=4.1cm, minimum height=2.6cm, rotate=120] (ELeft)  at ({9cm + \radius*cos(210)}, {\radius*sin(210)}) {};
\node[draw, ellipse, minimum width=4.1cm, minimum height=2.6cm, rotate=240] (ERight) at ({9cm + \radius*cos(330)}, {\radius*sin(330)}) {};

\draw [edgedot] (ETop) -> (ELeft);
\draw [edgedot] (ETop) -> (ERight);
\draw [edgedot] (ELeft) -> (ERight);

\node (GC11) at ({9cm + \radiussmall*cos(180) + \radius*cos(90)},  {\radius*sin(90)  + \radiussmall*sin(180)}) [box, minimum size=\r] {};
\node (GC12) at ({9cm - \radiussmall*cos(180) + \radius*cos(90)},  {\radius*sin(90)  - \radiussmall*sin(180)}) [box, minimum size=\r] {};
\node (GC13) at ({9cm + \radiussmall*cos(120) + \radius*cos(210)}, {\radius*sin(210) + \radiussmall*sin(120)}) [box, minimum size=\r] {};
\node (GC14) at ({9cm - \radiussmall*cos(120) + \radius*cos(210)}, {\radius*sin(210) - \radiussmall*sin(120)}) [box, minimum size=\r] {};
\node (GC15) at ({9cm + \radiussmall*cos(60)  + \radius*cos(330)}, {\radius*sin(330) + \radiussmall*sin(60)})  [box, minimum size=\r] {};
\node (GC16) at ({9cm - \radiussmall*cos(60)  + \radius*cos(330)}, {\radius*sin(330) - \radiussmall*sin(60)})  [box, minimum size=\r] {};

\draw [edgedashed] (GC11) -> (GC12);
\draw [edgedashed] (GC13) -> (GC14);
\draw [edgedashed] (GC15) -> (GC16);

\node (G1) at ({9cm + \radiussmall*cos(180) + \radius*cos(90) - 0.35cm},  {\radius*sin(90)  + \radiussmall*sin(180) + 0.35cm}) [box] {$G_{1}$};
\node (G2) at ({9cm + \radiussmall*cos(180) + \radius*cos(90) + 0.35cm},  {\radius*sin(90)  + \radiussmall*sin(180) + 0.35cm}) [box] {$G_{2}$};
\node (G3) at ({9cm + \radiussmall*cos(180) + \radius*cos(90) - 0.35cm},  {\radius*sin(90)  + \radiussmall*sin(180) - 0.35cm}) [box] {$G_{3}$};
\node (G4) at ({9cm + \radiussmall*cos(180) + \radius*cos(90) + 0.35cm},  {\radius*sin(90)  + \radiussmall*sin(180) - 0.35cm}) [box] {$G_{4}$};

\draw [default_edge] (G1) -> (G2);
\draw [default_edge] (G1) -> (G3);
\draw [default_edge] (G1) -> (G4);
\draw [default_edge] (G2) -> (G3);
\draw [default_edge] (G2) -> (G4);
\draw [default_edge] (G3) -> (G4);

\node (G5) at ({9cm - \radiussmall*cos(180) + \radius*cos(90) - 0.35cm},  {\radius*sin(90)  - \radiussmall*sin(180) + 0.35cm}) [box] {$G_{5}$};
\node (G6) at ({9cm - \radiussmall*cos(180) + \radius*cos(90) + 0.35cm},  {\radius*sin(90)  - \radiussmall*sin(180) + 0.35cm}) [box] {$G_{6}$};
\node (G7) at ({9cm - \radiussmall*cos(180) + \radius*cos(90) - 0.35cm},  {\radius*sin(90)  - \radiussmall*sin(180) - 0.35cm}) [box] {$G_{7}$};
\node (G8) at ({9cm - \radiussmall*cos(180) + \radius*cos(90) + 0.35cm},  {\radius*sin(90)  - \radiussmall*sin(180) - 0.35cm}) [box] {$G_{8}$};

\draw [default_edge] (G5) -> (G6);
\draw [default_edge] (G5) -> (G7);
\draw [default_edge] (G5) -> (G8);
\draw [default_edge] (G6) -> (G7);
\draw [default_edge] (G6) -> (G8);
\draw [default_edge] (G7) -> (G8);

\node (G9) at ({9cm + \radiussmall*cos(120) + \radius*cos(210) - 0.35cm},  {\radius*sin(210)  + \radiussmall*sin(120) + 0.35cm}) [box] {$G_{9}$};
\node (G10) at ({9cm + \radiussmall*cos(120) + \radius*cos(210) + 0.35cm},  {\radius*sin(210)  + \radiussmall*sin(120) + 0.35cm}) [box] {$G_{10}$};
\node (G11) at ({9cm + \radiussmall*cos(120) + \radius*cos(210) - 0.35cm},  {\radius*sin(210)  + \radiussmall*sin(120) - 0.35cm}) [box] {$G_{11}$};
\node (G12) at ({9cm + \radiussmall*cos(120) + \radius*cos(210) + 0.35cm},  {\radius*sin(210)  + \radiussmall*sin(120) - 0.35cm}) [box] {$G_{12}$};

\draw [default_edge] (G9) -> (G10);
\draw [default_edge] (G9) -> (G11);
\draw [default_edge] (G9) -> (G12);
\draw [default_edge] (G10) -> (G11);
\draw [default_edge] (G10) -> (G12);
\draw [default_edge] (G11) -> (G12);

\node (G13) at ({9cm - \radiussmall*cos(120) + \radius*cos(210) - 0.35cm},  {\radius*sin(210)  - \radiussmall*sin(120) + 0.35cm}) [box] {$G_{13}$};
\node (G14) at ({9cm - \radiussmall*cos(120) + \radius*cos(210) + 0.35cm},  {\radius*sin(210)  - \radiussmall*sin(120) + 0.35cm}) [box] {$G_{14}$};
\node (G15) at ({9cm - \radiussmall*cos(120) + \radius*cos(210) - 0.35cm},  {\radius*sin(210)  - \radiussmall*sin(120) - 0.35cm}) [box] {$G_{15}$};
\node (G16) at ({9cm - \radiussmall*cos(120) + \radius*cos(210) + 0.35cm},  {\radius*sin(210)  - \radiussmall*sin(120) - 0.35cm}) [box] {$G_{16}$};

\draw [default_edge] (G13) -> (G14);
\draw [default_edge] (G13) -> (G15);
\draw [default_edge] (G13) -> (G16);
\draw [default_edge] (G14) -> (G15);
\draw [default_edge] (G14) -> (G16);
\draw [default_edge] (G15) -> (G16);

\node (G17) at ({9cm + \radiussmall*cos(60) + \radius*cos(330) - 0.35cm},  {\radius*sin(330)  + \radiussmall*sin(60) + 0.35cm}) [box] {$G_{17}$};
\node (G18) at ({9cm + \radiussmall*cos(60) + \radius*cos(330) + 0.35cm},  {\radius*sin(330)  + \radiussmall*sin(60) + 0.35cm}) [box] {$G_{18}$};
\node (G19) at ({9cm + \radiussmall*cos(60) + \radius*cos(330) - 0.35cm},  {\radius*sin(330)  + \radiussmall*sin(60) - 0.35cm}) [box] {$G_{19}$};
\node (G20) at ({9cm + \radiussmall*cos(60) + \radius*cos(330) + 0.35cm},  {\radius*sin(330)  + \radiussmall*sin(60) - 0.35cm}) [box] {$G_{20}$};

\draw [default_edge] (G17) -> (G18);
\draw [default_edge] (G17) -> (G19);
\draw [default_edge] (G17) -> (G20);
\draw [default_edge] (G18) -> (G19);
\draw [default_edge] (G18) -> (G20);
\draw [default_edge] (G19) -> (G20);

\node (G21) at ({9cm - \radiussmall*cos(60) + \radius*cos(330) - 0.35cm},  {\radius*sin(330)  - \radiussmall*sin(60) + 0.35cm}) [box] {$G_{21}$};
\node (G22) at ({9cm - \radiussmall*cos(60) + \radius*cos(330) + 0.35cm},  {\radius*sin(330)  - \radiussmall*sin(60) + 0.35cm}) [box] {$G_{22}$};
\node (G23) at ({9cm - \radiussmall*cos(60) + \radius*cos(330) - 0.35cm},  {\radius*sin(330)  - \radiussmall*sin(60) - 0.35cm}) [box] {$G_{23}$};
\node (G24) at ({9cm - \radiussmall*cos(60) + \radius*cos(330) + 0.35cm},  {\radius*sin(330)  - \radiussmall*sin(60) - 0.35cm}) [box] {$G_{24}$};

\draw [default_edge] (G21) -> (G22);
\draw [default_edge] (G21) -> (G23);
\draw [default_edge] (G21) -> (G24);
\draw [default_edge] (G22) -> (G23);
\draw [default_edge] (G22) -> (G24);
\draw [default_edge] (G23) -> (G24);

\normalsize

\end{tikzpicture}

%% file: balanced_partition.tex
While GPMP asks for a mapping with minimal communication cost, it also requires that the $k$-way partition is $\epsilon$-balanced.
Recall that a partition $V = V_1 \cup V_2 \cup \ldots \cup V_k$ is $\epsilon$-balanced if $c(V_i) \leq L_{\max} \coloneqq \left \lceil(1 + \epsilon)\frac{c(V)}{k}\right \rceil$.
This means that each partition is allowed to weigh $(1+\epsilon)$ times the average partition weight.
The constraint is necessary in the context of process mapping, as the partitions represent the workloads distributed to the processing elements of real hardware.
If no constraint was present, then $V_1 = V$ and $V_i = \emptyset$ would be a perfect solution with no communication cost.
However, only one PE would be active as all others remain idle, which would lead to overall worse performance.
Enforcing the $\epsilon$-balance ensures that the work is distributed equally and all PEs are equally utilized.

To ensure that the final $k$-way partition is \hbox{$\epsilon$-balanced}, it is necessary to adaptively rescale $\epsilon$ for each partitioning.
For example, consider a graph with 800 vertices, each with weight one, the hierarchy \hbox{$H=4:2$}, $k = 8$ and $\epsilon=0.1$.
We assume the worst case, where one block always maximizes its allowed weight.
When partitioning the first graph, we obtain a subgraph with a weight of $(1+0.1)800/2 = 440$.
In the next step, this subgraph is partitioned and one of its blocks will receive a weight of $(1+0.1) 440/4 = 121$.
However, this block exceeds the maximum allowed weight of \hbox{$L_{\max} = (1 + 0.1)800 / 8 = 110$}, and the resulting partition is not $\epsilon$-balanced.

It is necessary to adaptively adjust the imbalance parameter $\epsilon'$ based on the current subgraph weight and its depth in the hierarchy.
The way to rescale $\epsilon$ and the proof we present here are similar to the ones in~\cite{Schlag16}.
\begin{lemma}[Adaptive Imbalance]
    Let $G=(V,E)$ be the graph to be partitioned, $\epsilon$ the allowed imbalance and $k = \prod_{i=1}^{\ell} a_i$ the number of partitions, with the $a_i$\('\)s describing the hierarchy.
    Let $G'=(V',E')$ be the subgraph to be partitioned, $d$ the depth in the hierarchy (where the original graph $G_{\mathcal{C}}$ has depth $\ell$ and the final subgraphs in the hierarchy have depth 0) and $k' = a_1 \cdot \ldots \cdot a_d$ be the number of partitions for $G'$.
    Using
    \begin{equation}
        \label{eq:imbalance}
        \epsilon' \coloneqq \Big((1 + \epsilon) \frac{k'c(V)}{k c(V')}\Big)^{\frac{1}{d}} - 1
    \end{equation}
    as the adaptive imbalance parameter to partition $G'$ into $a_d$ partitions ensures that the final $k$-way partition of $G$ is $\epsilon$-balanced.
\end{lemma}

\begin{proof}
(Outline)
    In the end, none of the $k$ partitions should have a weight greater than \hbox{$L_{\max} = (1 + \epsilon)\frac{c(V)}{k}$}.
    Assume we use the original imbalance parameter $\epsilon' = \epsilon$ and that there is always one partition $V_{\max}$ that receives the maximum possible weight.
    One final block would then have a weight of
    \begin{equation}
        \label{eq:uncorrect-weight}
        c(V_{\max}) = \frac{(1+\epsilon')}{a_1} \cdots \frac{(1+\epsilon')c(V)}{a_\ell} = \frac{(1+\epsilon')^\ell c(V)}{k}.
    \end{equation}
    To ensure the balancing of this block, we have to choose $\epsilon'$ such that $c(V_{\max}) \leq L_{\max}$.
    When partitioning a block $V_i$ at depth $d$ into $k'$ blocks, we choose $\epsilon'$ as follows:
    \begin{equation}
        \label{eq:balanced_epsilon}
        (1 + \epsilon')^{d}\frac{c(V_i)}{k'} \leq L_{\max} \rightarrow \epsilon' \leq \Big((1 + \epsilon) \frac{k'c(V)}{k c(V')}\Big)^{\frac{1}{d}} - 1
    \end{equation}
    and thereby ensure a final $\epsilon$-balanced $k$-way partition.
\end{proof}

%% file: experiments.tex
\subsection{Methodology.}\label{subsec:methodology}
Our algorithm is implemented in C++ Standard 17, utilizing C++ Standard Threads for our parallel algorithms.
To facilitate graph partitioning, we rely on two state-of-the-art libraries: \textsc{KaFFPa}~\cite{Sanders11} of the \textsc{KaHIP} library (version 3.16) as the serial graph partitioner and \textsc{Mt-KaHyPar}~\cite{Gottesbueren20} (version 1.4) as the shared-memory graph partitioner.
Note that \textsc{Mt-KaHyPar} uses Intel's TBB library internally for parallelism.
Both libraries are also written in C++ and are integrated via their interface functions.
The code and the libraries are compiled using GCC version 14.1.0 with full optimization (-O3 flag).

The experiments are conducted on a machine that fits two Intel Xeon Gold 6230, each with 20 cores and 40 threads, for a total of 80 threads.
The CPU frequency is 2.1~Ghz.
The main memory was capped at 175~GB, and the machine runs Red Hat Enterprise~8.8.

\subsection{Benchmark Instances.}\label{subsec:benchmark-instances}
\input{graphs}

The benchmark instances can be seen in Table~\ref{tab:benchmark-instances}.
We chose them as they are the same as in~\cite{Faraj20, Heuer23, Schulz17} and~\cite{Schulz19}, which makes for an easier comparison.
The instances come from various sources, with the small ones coming from the SuiteSparse Matrix Collection~\cite{Davis11} and most of the medium-sized graphs are taken from Chris Walshaw’s benchmark archive~\cite{Soper04}.
We also use graphs from the 10th DIMACS Implementation Challenge~\cite{Bader14} website.
The graphs \texttt{rgg23} and \texttt{rgg24} are random geometric graphs with $2^{23}$ and $2^{24}$ vertices, where each vertex represents a random point in the unit square and edges connect vertices whose Euclidean distance is below $0.55\sqrt{\ln{n} / n}$.
The graphs \texttt{del23} and \texttt{del24} are Delaunay triangulations of $2^{23}$ and $2^{24}$ random points in the unit square.
The graphs \texttt{af\_shell9}, \texttt{thermal2}, and \texttt{nlr} are from the matrix and the numeric section of the DIMACS benchmark set.
The graphs \texttt{eur} and \texttt{deu} are large road networks of Europe and Germany taken from~\cite{Delling09}.

\subsection{Results.}\label{subsec:results}
As the hierarchy we chose \hbox{$H = 4:8:\{1, \ldots, 6\}$} for all our experiments and for the distance \hbox{$D=1:10:100$}.
As the imbalance parameter, we choose an imbalance of 3\% i.e., $\epsilon = 0.03$.
We compare both \emph{communication cost} $J(C, D, \Pi)$ and \emph{running times} of our algorithms.
Each algorithm is run three times with different seeds, and both running times and $J(C, D, \Pi)$\('\)s are averaged.

Note that not all algorithms achieve an $\epsilon$-balanced partition for each configuration, graph, and seed.
For our algorithm $0.04\%$ of all instances (13 of 21\ 587) are imbalanced with a maximum imbalance of $4.1\%$.
Due to the negligible number of instances and the at most $1.1\%$ additional imbalance, we will not handle imbalanced partitions differently during the analysis.

\textbf{Performance Profiles.}
We use performance profiles~\cite{Dolan02} to compare the solution quality across the algorithms.
Let $\mathcal{A}$ be the set of all algorithms, $\mathcal{I}$ the set of all instances, $q_A(I)$ the quality of algorithm $A \in \mathcal{A}$ on instance $I \in \mathcal{I}$ and $\text{Best}(I) = \max_{A \in \mathcal{A}} q_A(I)$ the best solution on instance $I$.
For each algorithm $A$, a performance plot shows the fraction of instances ($y$-axis) for which $q_A(I) \leq \tau \cdot \text{Best}(I)$, where~$\tau$ is on the x-axis.
For~$\tau = 1$ the plot shows the fraction of instances, that each algorithm solved with the best solution.
Algorithms with greater fractions at small $\tau$ values (near the top-left corner) are better in terms of solution quality.

\begin{figure}[t]
    \centering
    \includegraphics[width=\columnwidth]{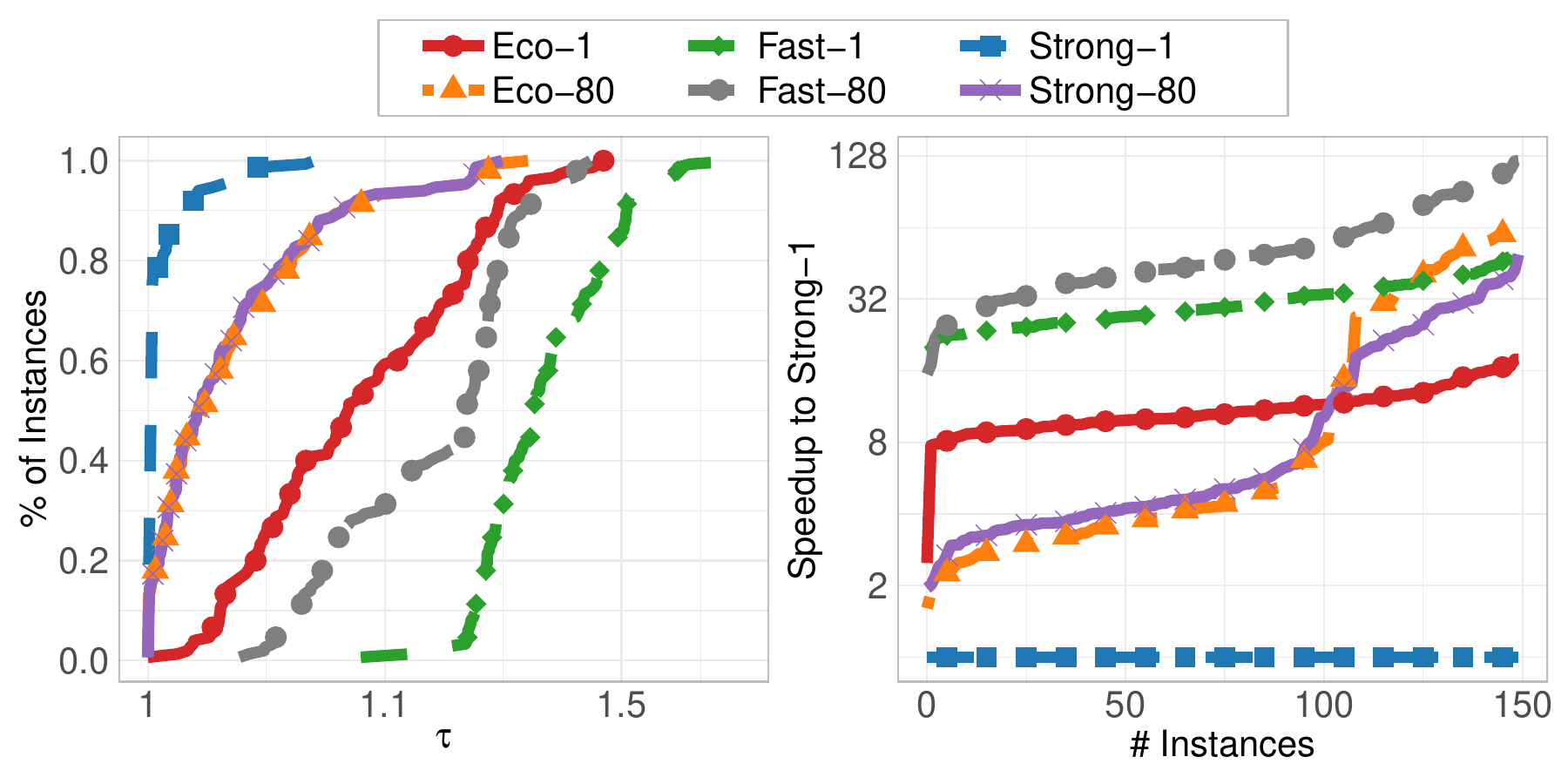}
    \caption{Solution quality (left) and speedup over \textsc{Strong-1} (right) for the 1 and 80 threaded \textsc{Non-Blocking Layer} \textsc{Fast}/\textsc{Eco}/\textsc{Strong} configurations.
    }
    \label{fig:FastEcoStrong}
\end{figure}

\textbf{Algorithm Configuration.}
Graph partitioning libraries typically provide a set of preconfigured hyperparameters that control the tradeoff between speed and quality.
In the case of \textsc{KaFFPa}, these are called \textsc{Fast}, \textsc{Eco} and \textsc{Strong}, and in the case of \textsc{Mt-KaHyPar} they are called \textsc{Default}, \textsc{Quality} and \textsc{Highest-Quality}.
Configurations with greater quality (i.e., \textsc{Strong} and \textsc{Highest-Quality}) usually enable more sophisticated local search methods that take longer, but can greatly improve the quality.

In our algorithm, it is possible to use a different configuration for each partitioning task.
Here, we settled for the simplest approach by specifying one serial configuration $A_{\text{ser}}$ and one parallel configuration $A_{\text{par}}$.
If at least two threads are assigned to a partitioning task then $A_{\text{par}}$ is used, otherwise $A_{\text{ser}}$ is used.
We recreate the \textsc{Fast}/\textsc{Eco}/\textsc{Strong} configurations by setting $A_{\text{ser}}$ to \textsc{Fast}/\textsc{Eco}/\textsc{Strong} and $A_{\text{par}}$ to \textsc{Default}/\textsc{Quality}/\textsc{Highest-Quality}.

Fig.~\ref{fig:FastEcoStrong} compares the solution quality and speedup using 1 and 80 threads.
As expected, \textsc{Strong} has a better solution quality than \textsc{Eco} and \textsc{Fast}, but is also slower.
\textsc{Strong} significantly outperforms \textsc{Eco} and \textsc{Fast} in the serial case, however, it is also up to 60 times slower than \textsc{Fast-1}.
When using 80 threads, the difference between \textsc{Strong} and \textsc{Eco} is not as substantial, neither in solution quality nor in speed.
\hbox{\textsc{Strong-1}} has significantly better solution quality than \hbox{\textsc{Strong-80}}, due to \textsc{KaFFPa-Strong} creating smaller edge-cut partitions than \textsc{Mt-KaHyPar-Highest-Quality}.
In comparison, \textsc{Mt-KaHyPar-Default} creates better partitions than \textsc{KaFFPa-Fast} leading to better solution quality for \textsc{Fast-80}.
However, \textsc{Fast-80} is not substantially faster than \textsc{Fast-1}.

\begin{figure}[t]
    \centering
    \includegraphics[width=\columnwidth]{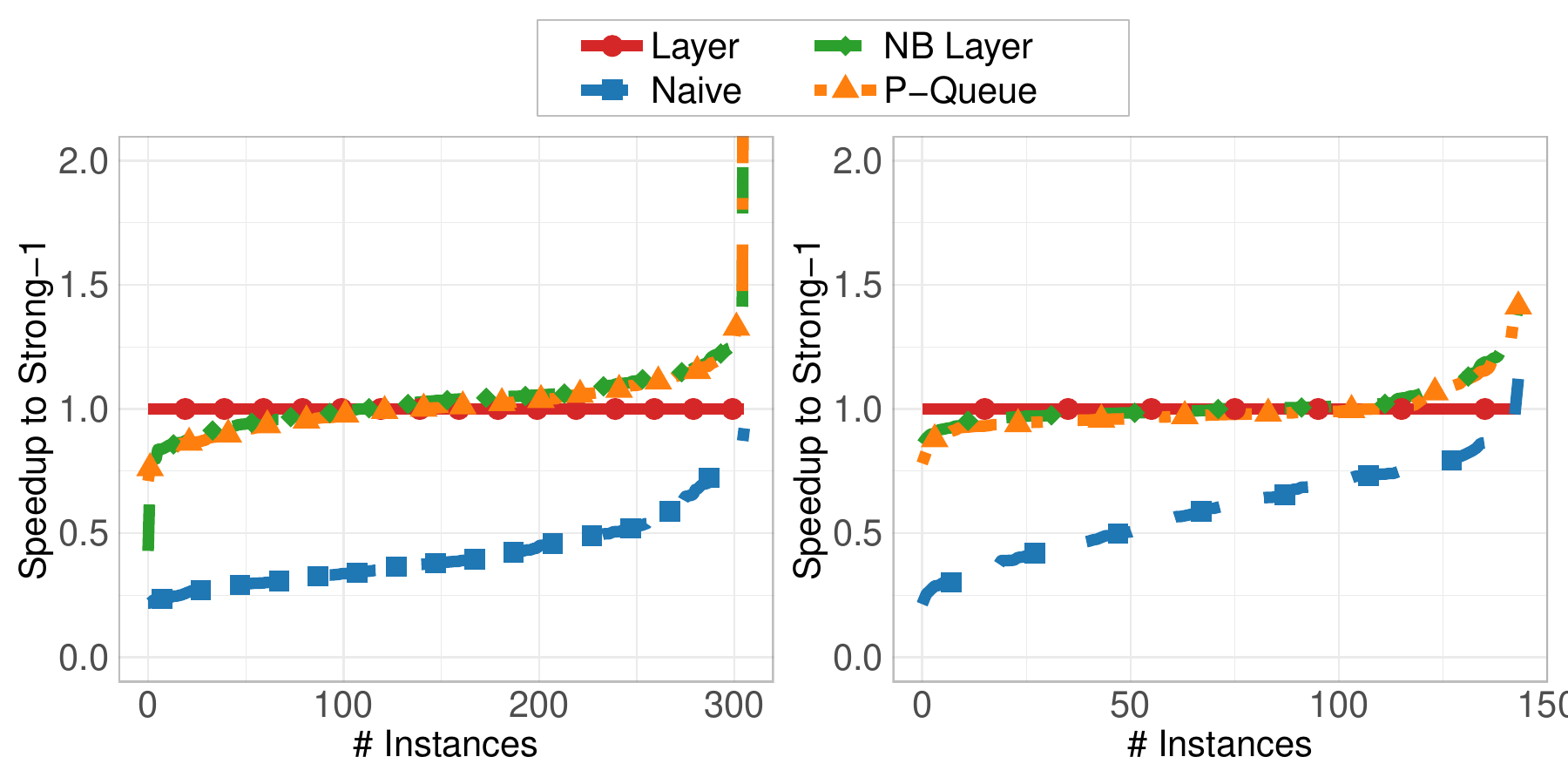}
    \caption{Runtime comparison for 80 threads between \textsc{Naive}, \textsc{Layer}, \textsc{Queue} and \textsc{Non-Blocking Layer} on small graphs (left) and large graphs (right).
    }
    \label{fig:approach}
\end{figure}

\textbf{Thread Distribution Approaches.}
Next, we compare the four thread distribution approaches presented in Section~\ref{sec:parallel-hierarchical-multisection}.
Fig.~\ref{fig:approach} shows the speedup of each approach compared to \textsc{Layer} when using 80 threads.
The left plot only includes instances with small graphs (less than one million vertices), while the right plot includes instances with large graphs (more than one million vertices).
\textsc{Naive} performs the worst, regardless of graph size.
There is no instance where it is the fastest approach.
For small graphs \textsc{Non-Blocking Layer} had the fastest running time on roughly ~49\% of instances, followed by \textsc{Layer} with ~28\% and \textsc{Priority Queue} with ~23\% of instances.
The geometric mean speedup for \textsc{Non-Blocking Layer} and \textsc{Priority Queue} over \textsc{Layer} on small graphs is $1.02$ and $1$ respectively.
\textsc{Layer} is the fastest approach most often with ~47\% of instances for larger graphs, followed by \textsc{Non-Blocking Layer} with ~42\% and \textsc{Priority Queue} with ~11\% of instances.
For large graphs, the geometric mean speedup is $1.01$ for \textsc{Non-Blocking Layer} and $0.99$ for \textsc{Priority Queue}.

In total, distributing the threads across the graphs is better than the \textsc{Naive} approach, however, which of the three methods is used has, on average, only a minimal effect on the runtime.

\begin{figure}[t]
    \centering
    \includegraphics[width=\columnwidth]{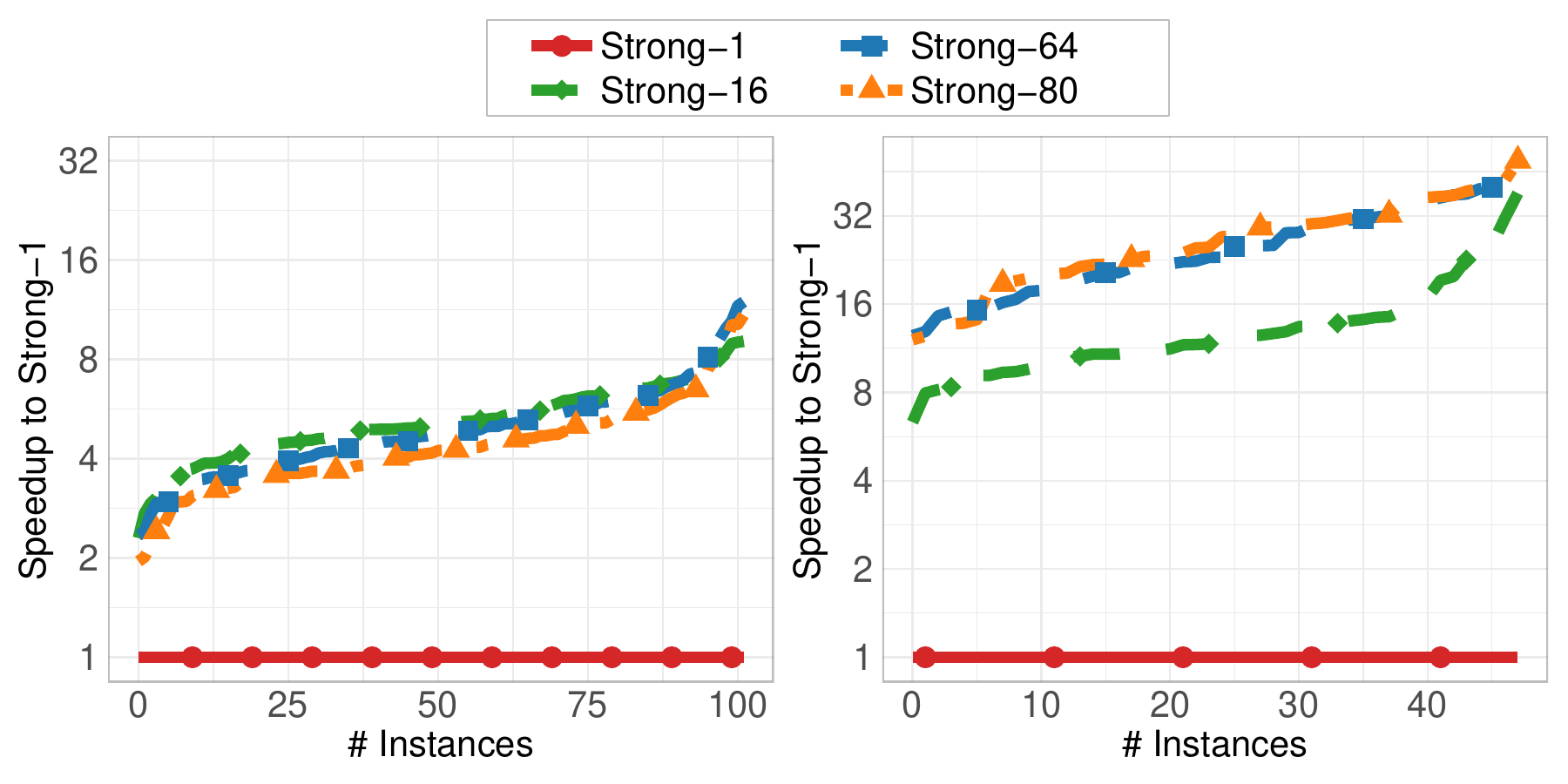}
    \caption{Comparing scalability of \textsc{Non-Blocking Layer} with the \textsc{Strong} configuration on small graphs (left) and large graphs (right).
    }
    \label{fig:scalability}
\end{figure}

\textbf{Scalability.}
Fig.~\ref{fig:scalability} shows the difference in speed when using 1, 16, 64 and 80 threads for the \textsc{Non-Blocking Layer} \textsc{Strong} configuration.
The left plot includes small graphs, while the right focuses on large graphs.
For the small graphs, the 16 threaded version can achieve a speedup of about 9.
The 64 threaded version and 80 threaded version achieve a maximum speedup of 11.9 and 10.8 respectively.
However, their geometric mean speedups are 5.1, 4.8 and 4.3.
The decrease in performance can be attributed to the small graphs, which are hard to efficiently partition in parallel.
For the large graphs, the 16, 64 and 80 threaded versions achieve a maximum speedup of 38.1, 49.1 and 49.2.
While using 64 threads instead of 16 has a notable impact on running time, the difference between using 64 and 80 threads is most often negligible.
The 16 threaded version achieves speedups of greater than 16 and therefore has superlinear speedup.
This is only possible since we compare to \textsc{Strong-1}, which exclusively uses \textsc{KaFFPa-Strong}.
The multithreaded versions also use \hbox{\textsc{Mt-KaHyPar-Highest-Quality}}, which is faster but has worse solution quality (see Fig.~\ref{fig:FastEcoStrong}).
Therefore, the speedups can not only be attributed to using more threads, but also to the use of another library when more threads are available.

\subsection{Comparison to SOTA.}\label{subsec:comparison-to-SOTA}
In this section, we compare our best algorithms to current state-of-the-art parallel, but also serial methods.
To better distinguish between the algorithms, we call our algorithm \textsc{SharedMap}, which will always use the \textsc{Non-Blocking Layer} approach.
Configurations \textsc{Strong} and \textsc{Fast} are abbreviated by \textsc{S} and \textsc{F} respectively.

\begin{figure}[t]
    \centering
    \includegraphics[width=\columnwidth]{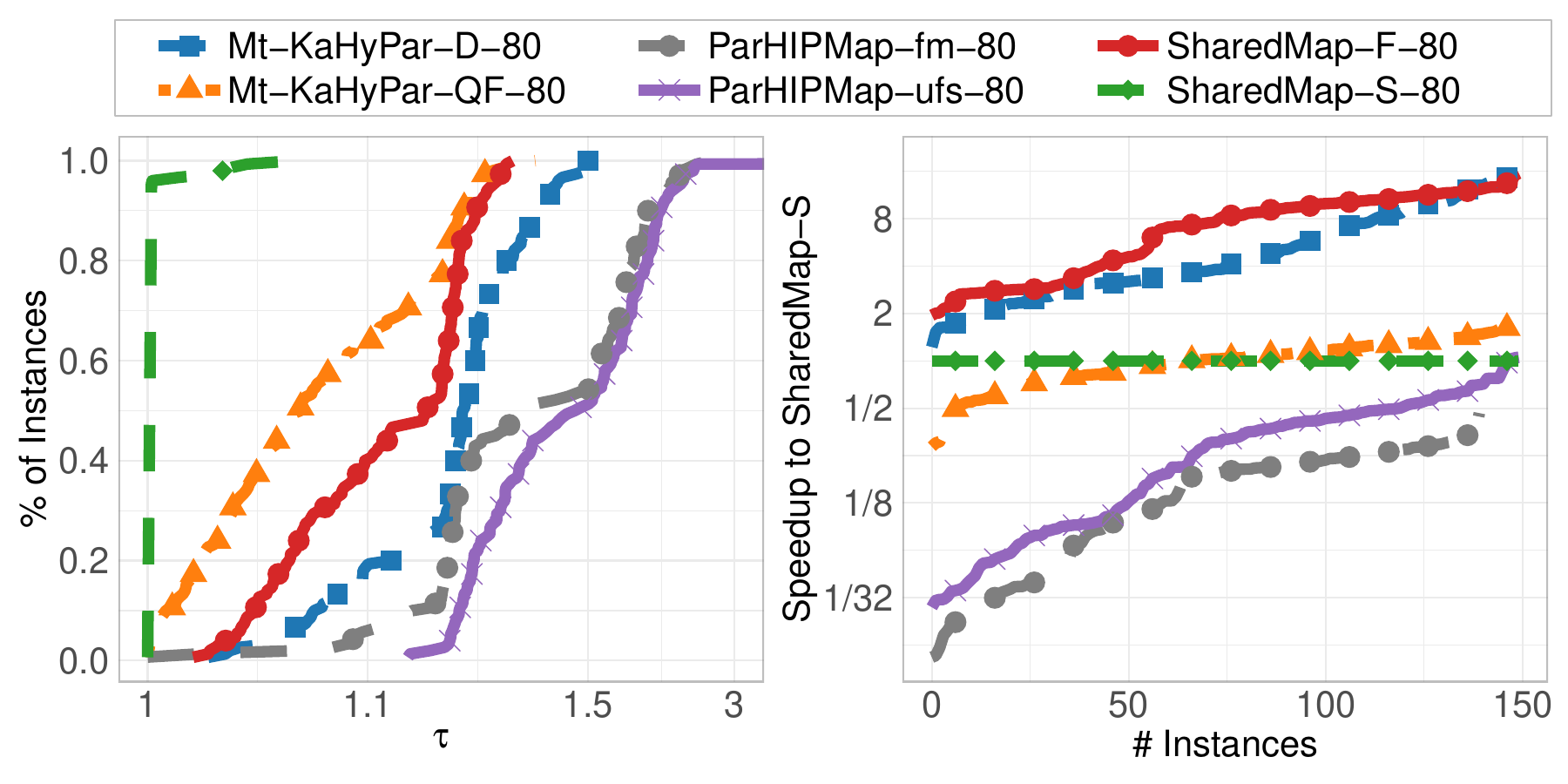}
    \caption{Solution quality (left) and speedup over \textsc{SharedMap-S} for the various parallel implementations. All algorithms are run with 80 threads.}
    \label{fig:parallel}
\end{figure}

\textbf{Parallel.}
We compare to \textsc{Mt-KaHyPar-Steiner}~\cite{Heuer23} as a parallel shared-memory algorithm and against \textsc{ParHIPMap}~\cite{Predari21} as a parallel distributed-memory algorithm.
For \textsc{Mt-KaHyPar-Steiner}, we use version 1.3\footnote{Version 1.4 gave inconsistent results.} and compare against its presets \textsc{Default} \hbox{(\textsc{Mt-KaHyPar-D})} and \textsc{Quality Flows} \hbox{(\textsc{Mt-KaHyPar-QF})}, which are the fastest and strongest configuration respectively.
For \textsc{ParHIPMap}, we evaluate the configurations \textsc{ultrafastsocial} (\textsc{ParHIPMap-ufs}) as the fastest configuration and \textsc{fastmesh} (\textsc{ParHIPMap-fm}) as the strongest configuration.
All algorithms utilize the 80 available threads.
Fig.~\ref{fig:parallel} shows the performance plot and speedup over \textsc{SharedMap-S}.
\textsc{SharedMap-S} significantly outperforms the other parallel algorithms in regard to solution quality.
It offers the best solution on 95\% on instances, while the next strongest algorithm, \textsc{Mt-KaHyPar-QF}, has the best solution on $5\%$ of instances.
While \hbox{\textsc{Mt-KaHyPar-QF}} can be up to $1.9$ times faster than \textsc{SharedMap-S}, its geometric mean speedup is slower at approximately $0.96$.
\textsc{SharedMap-F} has better solution quality than \textsc{Mt-KaHyPar-D} and additionally is also faster, with a geometric mean speedup of $6.4$ to $4.5$ over \textsc{SharedMap-S}.
The \textsc{ParHIPMap} algorithms can neither compete in speed nor solution quality.

Overall, \textsc{SharedMap-S} and \textsc{SharedMap-F} have better solution quality and are faster compared to the corresponding counter-parts of \textsc{Mt-KaHyPar}.

\textbf{Serial.}
In this section, all algorithms run serially, meaning that each algorithm only uses one thread.
Although this work focuses on a multithreaded algorithm, we also want to compare our algorithms to the three serial methods \hbox{\textsc{KaFFPa-Map}~\cite{Schulz17}}, \textsc{Global multisection}~\cite{Schulz19} (\textsc{GM}) and \hbox{\textsc{Integrated Mapping}~\cite{Faraj20}} (\textsc{IM}) of the \textsc{KaHIP} library.
We compare to their \textsc{Strong} and \textsc{Fast} configurations.
\textsc{SharedMap} will therefore only utilize the \textsc{KaFFPa} partitionings of $A_{\text{ser}}$ and the thread distribution strategy is irrelevant.
Note that \textsc{Jostle}~\cite{Karypis98} and \textsc{Scotch}~\cite{Pellegrini96} also offer methods to solve GPMP, however,~\cite{Faraj20} and~\cite{Schulz19} report that the mapping algorithms in \textsc{KaHIP} outperform these methods.
Also, a comparison with \textsc{LibTopoMap}~\cite{Hoefler11} is omitted since \textsc{GM}~\cite{Schulz19} outperforms the algorithm.

\begin{figure}[t]
    \centering
    \includegraphics[width=\columnwidth]{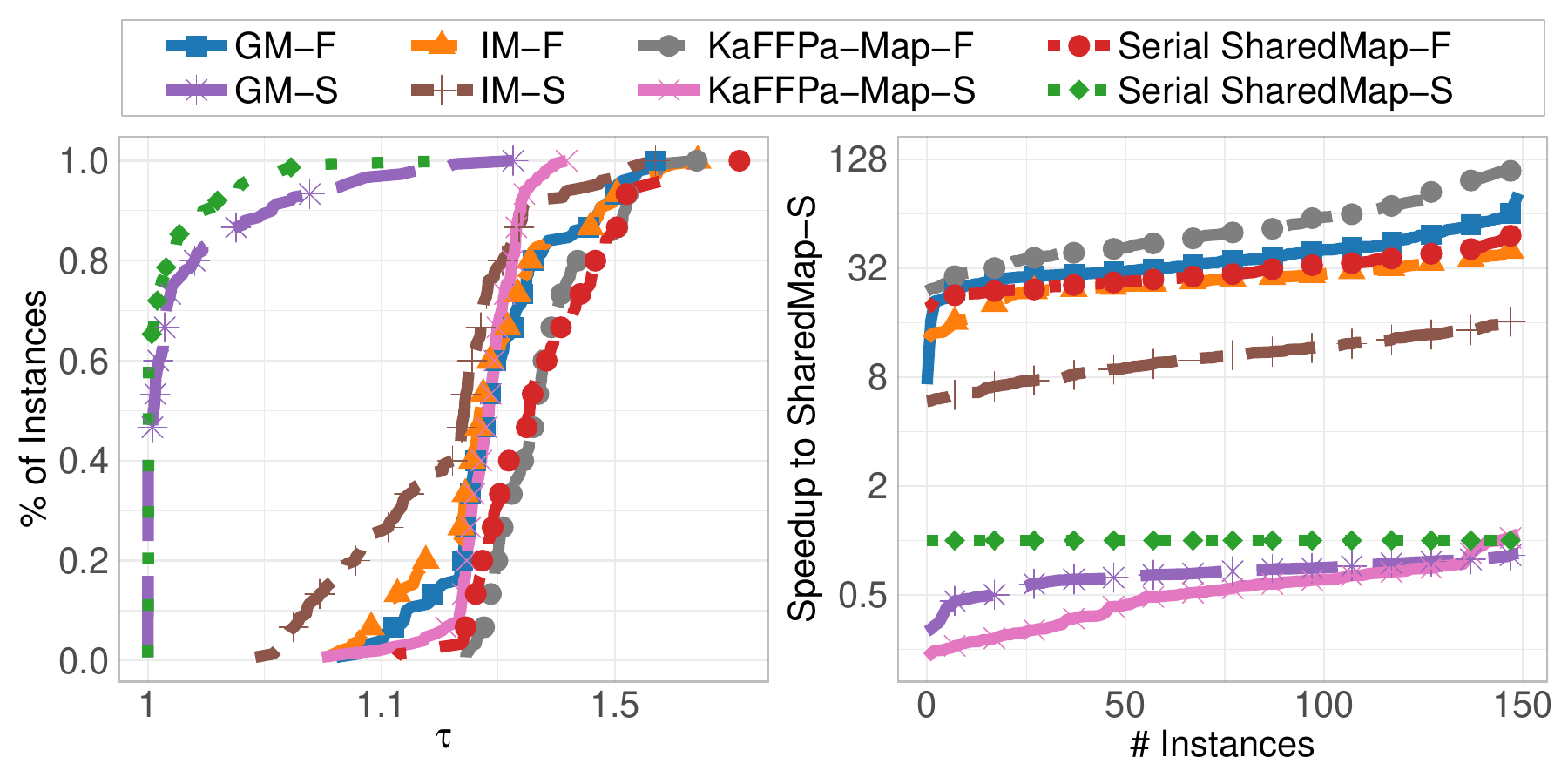}
    \caption{Solution quality (left) and speedup over serial \textsc{SharedMap-S} (right) for the various other serial implementations.}
    \label{fig:serial}
\end{figure}

Fig.~\ref{fig:serial} shows the performance plot and speedup compared to \textsc{SharedMap-S}.
\textsc{SharedMap-S} ($60\%$ best solutions) and \textsc{GM-S} ($40\%$ best solutions) are the best in terms of solution quality.
\textsc{SharedMap-S} is faster than \textsc{GM-S}, which has a geometric mean speedup of $0.64$.
Both \textsc{SharedMap-S} and \textsc{GM-S} use hierarchical multisection and \textsc{Strong} graph partitioning from the \textsc{KaHIP} library, which raises the question of why our implementation is stronger and faster.
To our knowledge, \textsc{GM} does not use the adaptive imbalance described in Section~\ref{sec:balanced-partition}.
This could enable our algorithm to find better partitions.
\textsc{GM-S} worse performance could be explained by the additional refinement that takes place.
They use a local search that swaps the assignment of two tasks if it yields a better mapping.
Our algorithm does not employ any further local search strategies.
\textsc{KaFFPa-Map-S} and \textsc{IM-S} cannot compete in terms of solution quality. %
\hbox{\textsc{SharedMap-F}} is one of the slowest out of all \textsc{Fast} configurations and also has one of the worst solution quality.
\hbox{\textsc{KaFFPa-Map-F}} has equally bad solutions, but it is the fastest of all algorithms.

In summary, serial \hbox{\textsc{SharedMap-S}} improves state-of-the-art results, while also being faster than the previously strongest algorithm.

%% file: graphs.tex
\begin{table}[t]
    \caption{Benchmark instance properties.}
    \begin{center}
        \begin{tabular}{lrr}
            \hline
            \multicolumn{1}{c}{Graph}          & \multicolumn{1}{c}{$|V|$}           & \multicolumn{1}{c}{$|E|$}           \\
            \hline
            \multicolumn{3}{c}{SuiteSparse Matrix Collection} \\
            \hline
            cop20k\_A      & 99 843         & 1 262 244       \\
            2cubes\_sphere & 101 492        & 772 886         \\
            thermomech\_TC & 102 158        & 304 700         \\
            cfd2           & 123 440        & 1 482 229       \\
            boneS01        & 127 224        & 3 293 964       \\
            Dubcova3       & 146 689        & 1 744 980       \\
            bmwcra\_1      & 148 770        & 5 247 616       \\
            G2\_circuit    & 150 102        & 288 286         \\
            shipsec5       & 179 860        & 4 966 618       \\
            cont-300       & 180 895        & 448 799         \\
            \hline
            \multicolumn{3}{c}{Walshaws' Benchmark Archive} \\
            \hline
            598a           & 110 971        & 741 934         \\
            fe\_ocean      & 143 437        & 409 593         \\
            144            & 144 649        & 1 074 393       \\
            wave           & 156 317        & 1 059 331       \\
            m14b           & 214 765        & 1 679 018       \\
            auto           & 448 695        & 3 314 611       \\
            \hline
            \multicolumn{3}{c}{Other Graphs} \\
            \hline
            af\_shell9     & $\approx504$K  & $\approx8.5$M   \\
            thermal2       & $\approx1.2$M  & $\approx3.7$M   \\
            nlr            & $\approx4.2$M  & $\approx12.5$M  \\
            deu            & $\approx4.4$M  & $\approx5.5$M   \\
            del23          & $\approx8.4$M  & $\approx25.2$M  \\
            rgg23          & $\approx8.4$M  & $\approx63.5$M  \\
            del24          & $\approx16.7$M & $\approx50.3$M  \\
            rgg24          & $\approx16.7$M & $\approx132.6$M \\
            eur            & $\approx18.0$M & $\approx22.2$M  \\
            \hline
        \end{tabular}
        \label{tab:benchmark-instances}
    \end{center}
\end{table}

%% file: conclusion.tex
Process mapping is the problem of mapping vertices of a task graph onto the processing elements of a supercomputer, such that the workload is equally distributed and communication cost is minimized.
In this work, we introduced a shared-memory hierarchical multisection algorithm that partitions the task graph alongside a homogeneous hardware hierarchy.
This strategy naturally creates multiple independent partitioning problems, enabling parallelization.
We described four different approaches for assigning threads to the partitionings, such that as few threads as possible are idle at all times.

Our algorithm \textsc{SharedMap} significantly outperforms current state-of-the-art parallel shared-memory algorithms in terms of solution quality and speed.
It has the best solution quality on $95\%$ of instances and is also faster than the previous best shared-memory algorithm.
In the serial case, \textsc{SharedMap-S} has the best quality and is faster than the previous state-of-the-art.

In the future, we want to explore hierarchical multisection on graphics processing units.
Another promising approach is to incorporate the mapping into the partitioning, as in~\cite{Faraj20}.
A shared-memory algorithm would be closely related to multithreaded graph partitioning algorithms and could significantly improve running times.
This would also include a multithreaded refinement phase, which could also be used to further improve the solution quality of \textsc{SharedMap}.

Although optimizing $J(C, D, \Pi)$ offers improvement in practice~\cite{Brandfass12}, it remains a relatively simple metric.
In the future, we aim to refine the model to more accurately represent real supercomputer architectures.

Modern large-scale scientific software is no longer static, so updates to the task graph can occur.
These changes require adaptive mapping strategies to maintain optimal efficiency.